\documentclass[twocolumn,journal]{IEEEtran}

\pdfoutput=1 

\IEEEoverridecommandlockouts                              
\usepackage{amsmath,amsfonts,amssymb}
\usepackage{cite,graphicx,url}
\usepackage{color}
\usepackage[ruled,commentsnumbered]{algorithm2e}
\usepackage{enumerate}

\graphicspath{{figures/}}

\definecolor{black}{rgb}{0,0,0}

\newcommand{\beq}{\begin{equation}}
\newcommand{\eeq}{\end{equation}}
\newcommand{\beqar}{\begin{eqnarray}}
\newcommand{\eeqar}{\end{eqnarray}}
\newcommand{\beqarno}{\begin{eqnarray*}}
\newcommand{\eeqarno}{\end{eqnarray*}}
\newcommand{\st}{\mathrm{s.t.}}

\newcommand{\umin}{u_{\rm min}}
\newcommand{\umax}{u_{\rm max}}
\newcommand{\ymin}{y_{\rm min}}
\newcommand{\ymax}{y_{\rm max}}
\newcommand{\emin}{e_{\rm min}}
\newcommand{\emax}{e_{\rm max}}

\newcommand{\eqdef}{\triangleq}
\newcommand{\rr}{\mathbb R}
\newcommand{\N}{\mathbb N}
\newcommand{\matrice}[2]{\left[\hspace*{-.1cm}\ba{#1} #2 \ea\hspace*{-.1cm}\right]}
\newcommand{\smallmat}[1]{\left[ \begin{smallmatrix}#1 \end{smallmatrix} \right]}

\newcommand{\ba}[1]{\begin{array}{#1}}
\newcommand{\ea}{\end{array}}

\newcommand{\Ss}{\mathcal{S}}

\newcommand{\RR}{\mathcal{R}}
\newcommand{\XX}{\mathcal{X}}

\newcommand{\fundef}[2]{:\rr^{#1}\rightarrow\rr^{#2}}

\newtheorem{theorem}{Theorem}
\newtheorem{lemma}[theorem]{Lemma}
\newtheorem{definition}[theorem]{Definition}

\newtheorem{proposition}[theorem]{Proposition}

\title{\LARGE \bf A Convex Feasibility Approach to Anytime Model Predictive Control}

\author{Alberto Bemporad, Daniele Bernardini, Panagiotis Patrinos
\thanks{The authors are with IMT Institute for Advanced Studies Lucca, Piazza San Francesco 19, 55100 Lucca, Italy. E-mail: {\tt\{alberto.bemporad,
daniele.bernardini,panagiotis.patrinos\}@imtlucca.it}.}%
}

\begin{document}

\maketitle
\thispagestyle{empty}
\pagestyle{empty}
\begin{abstract}
%

This paper proposes to decouple performance optimization 
and enforcement of asymptotic convergence in Model Predictive Control (MPC)
so that convergence to a given terminal set is achieved
independently of how much performance is optimized at each sampling step. 
By embedding an explicit decreasing condition in the MPC constraints
and thanks to a novel and very easy-to-implement 
convex feasibility solver proposed in the paper, it is possible
to run an outer performance optimization algorithm on top of the feasibility solver
and optimize for an amount of time that depends on the available CPU resources within 
the current sampling step (possibly going open-loop at a given sampling step 
in the extreme case no resources are available) and still guarantee convergence to the terminal set. While the MPC setup and the solver
proposed in the paper can deal with quite general classes of functions, we highlight
the synthesis method and show numerical results in case of linear MPC and ellipsoidal and polyhedral terminal sets.
\end{abstract}

\section{Introduction}
Model Predictive Control (MPC) is a well known advanced control approach in industry for
its capability of optimizing closed-loop performance subject to operating constraints on 
input and output variables~\cite{MR09,Mac02,Bem06}. In recent years, MPC has become
very attractive also in fast-sampling applications with stringent real-time requirements,
such as those arising in the automotive and aerospace industries. Such requirements
posed a research challenge for developing optimization algorithms, and in particular Quadratic Programming (QP) solvers, that enable the use of MPC in commercial products. In particular,
an embedded optimization solver must be fast, simple to code and test, require little memory, and have good worst-case estimates of its execution time.

To cope with such requirements, multiparametric QP was proposed in~\cite{BMDP02a} to pre-solve the QP off-line, therefore converting the MPC law into a continuous and piecewise affine function of the state vector. The main drawback of explicit MPC is that it is limited to relatively small problems and to linear time-invariant (LTI) systems.

On-line optimization methods like active-set methods~\cite{Ric85,QPOASES,SB94}, interior-point methods~\cite{MB10,WB10,ZRDMJ14}, and dual piecewise smooth Newton methods~\cite{PSS11} can be very effective in speed, but their worst-case CPU time can be hard to estimate in a non-conservative way. For accelerated dual gradient-projection methods~\cite{Nes83}, good bounds on the worst-case execution time were provided~\cite{PB14,RPB13}, although the methods act on the dual QP problem, and therefore can lead to infeasible solutions if the execution is interrupted.

On the other hand, in real-time control platforms the time allotted for the MPC controller to run is often not enough to cover the worst-case execution time, and other higher-priority tasks may even preempt its full execution. Driven by such real-time constraints, \emph{anytime control algorithms} were developed in~\cite{FGB08} with the idea of storing a set of control laws, each one of different complexity and closed-loop performance, and execute the one whose complexity is compatible with the current available CPU resources.

In this paper we propose instead an MPC approach based on \emph{anytime optimization}, with a novel convex optimization algorithm that recursively finds feasible solutions of decreasing level of suboptimality, depending on the computation power available within the sampling step. We first prove a rather general recursive feasibility and convergence result of MPC based on stability constraints that artificially impose a certain Lyapunov function to be decreasing~\cite{Bem98rhl,BB12}, where this function might be totally decoupled from the value function typically considered for assessing asymptotic convergence~\cite{SMR99,ZRDMJ14}. Moreover, in this paper we focus on \emph{convergence to a set} around an equilibrium rather than to an equilibrium state, as in practical applications is often sufficient to track a set-point within a given tolerance~\cite{DB13}. In addition, contrarily to~\cite{ZRDMJ14} that guarantees feasibility in real-time through a warm-starting technique in combination with robust MPC design, we provide an approach based on an original method for solving unconstrained problems, which is used for finding a feasible solution to a set of convex constraints. This method is very efficient in speed and easy to code, and can be run multiple times to approach an optimal solution, depending on the available CPU time.

The paper is organized as follows. Section~\ref{sec:setup} defines the main MPC setup and states recursive feasibility and convergence results. Section~\ref{sec:algo} presents the new convex feasibility and optimization algorithm setup and shows its properties. Section~\ref{sec:terminal-set} proposes two ways of synthesizing a proper terminal set and stability constraints for linear systems subject to linear constraints on inputs and outputs, and Section~\ref{sec:simulations} shows numerical evidence of the advantages of the proposed approach.

\subsection{Notation}
The sets of real and nonnegative integer numbers are denoted by $\rr$, $\N$, respectively.
For a vector $x\in\rr^{n}$, $x_i$ denotes the $i$-th entry of $x$, and the expression $x>0$ means that $x_i>0$, $\forall i=1,\ldots,n$. For a matrix $A\in\rr^{n\times m}$, $A_i$ denotes the $i$-th row of $A$, 
and $A>0$ positive definiteness of $A$. Given a scalar $\tau$, $\tau_+$ denotes $\max\{\tau,0\}$;
for a vector $x\in\rr^m$, $x_+$ is the vector whose coordinates are $(x_+)_i=\max\{x_i,0\}$, 
$\forall i=1\ldots,m$.

\section{Feasibility-based MPC}
\label{sec:setup} 
Consider the problem of steering the system
\begin{equation}
	x(t+1)=\alpha(x(t),u(t))
\label{eq:NLprocess}
\end{equation}
to a target set $\Ss\subseteq\rr^n$ while satisfying
the constraints
\begin{equation}
g(x(t),u(t))\leq 0
\label{eq:constraints}
\end{equation}
for all $t\in\N$, where $x\in\rr^n$, $u\in\rr^m$,
$\alpha\fundef{n+m}{}$ and $g\fundef{n+m}{}$. We represent the target set as
\begin{equation}
	\Ss\eqdef\{x\in\rr^n:\ f(x)\leq 0\}
\label{eq:target-set}
\end{equation}
where $f\fundef{n}{}$ and $\Ss$ is constrained controlled
invariant with respect to~\eqref{eq:constraints}, in accordance with the 
following definition.

\begin{definition}
A set $\Ss\subseteq\rr^n$ is \emph{constrained controlled invariant} 
if for all $x\in\Ss$ there exists $u\in\rr^m$ such that $g(x,u)\leq 0$, $\alpha(x,u)\in\Ss$.
\end{definition}

Note that in~\eqref{eq:constraints}--\eqref{eq:target-set} we are assuming scalar functions~$f$,~$g$ without loss of generality. In fact, for any vector function $h\fundef{n_1}{n_2}$, the component-wise constraint $h(x)\leq 0$ is equivalent to the constraint $\bar h(x)\eqdef\max_{i=1,\ldots,n_2}h_i(x)\leq 0$.
Note also that, given any set $\Ss\subseteq\rr^n$, a corresponding function $f\fundef{n}{}$ satisfying~\eqref{eq:target-set} can be defined as $f(x)=\gamma_{\Ss}(x)-1$, where $\gamma_{\Ss}(x)=\inf\{\mu: \mu\geq0,\, x\in\mu\Ss\}$ is the Minkowski function of $\Ss$. 


To solve the stated control problem, we consider the 
following MPC formulation 
\begin{subequations}
\begin{IEEEeqnarray}{lcl}
\min &\ & \ell_N(x_N)+\sum_{k=0}^{N-1}\ell_k(x_k,u_k) \label{eq:cost}\\
\st  &\ & x_0=x(t)\\
     &\ & \ x_{k+1}=\alpha(x_k,u_k),\  k=0,1,\ldots,N-1\label{eq:model}\\
     &\ & g(x_k,u_k)\leq 0,\ k=0,1,\ldots,N-1\label{eq:stage-constraint}\\
     &\ & f(x_N)\leq 0\label{eq:term-constraint}\\
     &\ & \sum_{k=1}^{N-1} f(x_k)_+\leq \phi(t-1)\label{eq:stab-constraint}
\end{IEEEeqnarray}%
\label{eq:MPC-problem}%
\end{subequations}
where $\ell_k\fundef{n+m}{}$, $k=0,\ldots,N-1$, $\ell_N\fundef{n}{}$ are stage and terminal costs, respectively, and $\phi(t)$ is a given scalar, chosen
in accordance with the following theorem.

\begin{theorem}
\label{th:convergence}
Let $u(t)=u_0^{t}$ be the control input applied to the process~\eqref{eq:NLprocess}, where $\{u_k^t\}_{k=0}^{N-1}$ is \emph{any feasible solution} of problem~\eqref{eq:MPC-problem} at time $t$, and the quantity
\beq
	\phi(t-1)\eqdef\sum_{k=1}^{N-1} f(x^{t-1}_k)_+
\eeq
is constructed from the previous feasible solution $u^{t-1}_k$, $x^{t-1}_k$ of problem~\eqref{eq:MPC-problem} at time $t-1$, for all $t\in\N$. If the set $\Ss$
defined in~\eqref{eq:target-set} is constrained controlled invariant and
problem~\eqref{eq:MPC-problem} is feasible at time $t=0$ for the initial state $x(0)$ and some value $\phi(-1)$, then it is feasible at all time $t\in\N$ and $x(t)\rightarrow\Ss$ for $t\rightarrow\infty$.
\end{theorem}

\proof
Let $\{u_k^t\}_{k=0}^{N-1}$ be any feasible solution of problem~\eqref{eq:MPC-problem} chosen at time $t$, and let $\{x_k^t\}_{k=0}^{N}$ be the corresponding state
trajectory. Consider the following candidate feasible solution 
$\{\bar u_k\}_{k=0}^{N-1}$ at time $t+1$, where
$\bar u_0=u_1^{t}$, $\bar u_1=u_2^{t}$, $\ldots$, $\bar u_{N-2}=u_{N-1}^{t}$,
and $\bar u_{N-1}$ such that $g(x_N^t,\bar u_{N-1})\leq 0$,
which exists by constrained controlled invariance of $\Ss$.
Let $\{\bar x_k\}_{k=0}^{N}$ the state trajectory corresponding 
to $\{\bar u_k\}_{k=0}^{N-1}$, with $\bar x_0=x(t+1)$. 
By construction, $\bar x_k=x_{k+1}^t$ for all $k=0,\ldots,N-1$, and
hence $g(\bar x_k,\bar u_k)\leq 0$ for all $k=0,\ldots,N-2$,
$f(\bar x_{N-1})\leq 0$. Moreover, by the choice of $\bar u_{N-1}$
we have $g(\bar x_{N-1},\bar u_{N-1})\leq 0$ and $f(\bar x_N)\leq 0$. 
Since $u(t)=u_0(t)$, and hence $x_1^t=x(t+1)$, we have
\begin{equation}
	\sum_{k=1}^{N-1} f(\bar x_k)_+=\sum_{k=2}^{N-1} f(x_k^t)_+=\phi(t)-f(x(t+1))_+ 
\label{eq:decreasing-phi}
\end{equation}
so that, since $f(x(t+1))_+\geq 0$, also the stability constraint~\eqref{eq:stab-constraint} is satisfied.
Therefore, problem~\eqref{eq:MPC-problem} admits a feasible solution at time $t+1$,
and because of~\eqref{eq:decreasing-phi}, whatever is the choice of 
$\{u_k^{t+1}\}_{k=0}^{N-1}$ at time $t+1$, we have
\begin{equation}
	\phi(t+1)\leq\phi(t)-f(x(t+1))_+
\label{eq:decreasing-phi2}
\end{equation}
This proves that $\lim_{t\rightarrow\infty}\phi(t)$ exists, as $\phi$ is a monotonically decreasing sequence and lower-bounded by zero, which in turns implies by~\eqref{eq:decreasing-phi2} that $\lim_{t\rightarrow\infty}f(x(t))_+=0$.
If by contradiction we assume that $x(t)\not\rightarrow\Ss$ for $t\rightarrow\infty$,
then a subsequence $t_h\in\N$, $h\in\N$, and a scalar $\delta>0$ exist such that $f(x(t_h))\geq \delta$, $\forall h\in\N$, or equivalently $f(x(t_h))_+\geq \delta$, which contradicts
$\lim_{t\rightarrow\infty}f(x(t))_+=0$.
\endproof
Note that the convergence result of Theorem~\ref{th:convergence} does not involve
at all the cost function~\eqref{eq:cost}. Of course, the transient behaviour of
the system depends on how close to optimality are the chosen 
feasible solutions $\{u_k^t\}_{k=0}^{N-1}$ of~\eqref{eq:MPC-problem}.

While the result of Theorem~\eqref{th:convergence} does not make
any assumption on the properties of functions $\alpha$, $f$, $g$, $\ell_k$ (except for constrained controlled invariance of $\Ss$), from now on we will restrict our attention to \emph{convex} functions $f$, $g$, $\ell_k$ and \emph{linear} functions $\alpha$, i.e., linear models
\begin{equation}
	x(t+1)=Ax(t)+Bu(t)
\label{eq:process}
\end{equation}
in order to solve problem~\eqref{eq:MPC-problem} effectively, using the novel algorithm proposed in the next section.

\section{Convex Feasibility Algorithm}
\label{sec:algo}

Consider the following feasibility problem:
\begin{equation}\label{eq:Feas}
\textrm{Find } x\in C\eqdef\{x\in\rr^n:\ f_i(x)\leq 0,\ i=1,\ldots,m\},
\end{equation}
where $f_i:\rr^n\to\rr$, $i=1,\ldots,m$ are convex, twice continuously differentiable functions. Problem~\eqref{eq:Feas} can be reformulated as the following unconstrained minimization problem:
\begin{equation}\label{eq:FeasMin}
\min\ F(x)=\frac{1}{2}\|f(x)_+\|_2^2
\end{equation}
where $f(x)=[f_1(x)\ \ldots\ f_m(x)]$. 
\begin{proposition}\label{prop:Suff}
If $C\neq\emptyset$, then $\inf F=0$ and $\arg\min F=C$. 
\end{proposition}
Clearly, if $\inf F>0$ then $C$ is empty.

\begin{lemma}
Function $F(x)=\frac{1}{2}\|f(x)_+\|_2^2$ is convex and continuously differentiable with
\begin{equation}\label{eq:gradF}
\nabla F(x)=\nabla f(x)'f(x)_+=\sum_{i=1}^mf_i(x)_+\nabla f_i(x)
\end{equation}
\end{lemma}
\begin{proof}
Function $F$ can be written as 
$$F(x)=\frac{1}{2}\sum_{i=1}^m(\max\{f_i(x),0\})^2=\frac{1}{2}\sum_{i=1}^m q(\psi_i(x)),$$
 where $\rr\ni z\mapsto q(z)=z^2$ and $\rr\ni z\mapsto \psi_i(x)=\max\{f_i(x),0\}$, $q\circ\psi$ is convex since $q$ is convex and nondecreasing when its argument is nonnegative and $\psi_i$ is convex (as the pointwise maximum of the convex function $f_i$ and $0$) and nonnegative. Therefore $F$ is convex as the sum of convex functions. On the other hand, $F(x)=\sum_{i=1}^m \varphi(f_i(x))$ where  $\rr\ni z\mapsto \varphi(z)=\frac{1}{2}(z)_+^2$. Since $\varphi$ is continuously differentiable with $\frac{d\varphi(z)}{dz}=(z)_+$, continuous differentiability of $F$ and formula~\eqref{eq:gradF} readily follow. 
\end{proof}
Our ultimate goal is to devise Newton-like methods for solving the unconstrained problem~\eqref{eq:FeasMin} and thus~\eqref{eq:Feas}. Function $F$ is $\mathcal{C}^1$ but not $\mathcal{C}^2$, however its gradient $\nabla F(x)$ is a piecewise smooth mapping, in the sense that for any $x\in\rr^n$ we have
\begin{equation}
\nabla F(x)\in\{\nabla F_I(x)\}_{I\in\mathcal{I}},
\label{eq:Nabla_F}
\end{equation}
where $\mathcal{I}$ is the collection of all subsets $I\subseteq\{1,\ldots,m\}$ for which there exists a $x\in\rr^n$ such that $f_i(x)\geq 0$, $i\in I$ and $f_i(x)< 0$, $i\notin I$, whereas $\nabla F_I(x)=\sum_{i\in I}f_i(x)\nabla f_i(x)$. 
The \emph{pieces} of $\nabla F$ are smooth with Jacobian given by
\begin{equation}
\nabla^2 F_I(x)=\sum_{i\in I}\nabla f_i(x)\nabla f_i(x)'+ f_i(x)\nabla^2 f_i(x).
\label{eq:Nabla2_F}
\end{equation}
For any $x\in\rr^n$ let $I(x)=\{i\in\{1,\ldots,m\}: f_i(x)\geq 0\}$. Then the matrix $\nabla^2F_{I(x)}$ will serve as a generalized Hessian of $F$ at $x$, furnishing a second-order approximation similar to the one provided by the classical Hessian for $\mathcal{C}^2$ functions. Another idea, stemming from Gauss-Newton methods for solving least-squares problems, would be to use as a generalized Hessian the matrix $\sum_{i\in I(x)}\nabla f_i(x)\nabla f_i(x)'$,
which results by omitting second-order terms $\nabla^2 f_i(x) f_i(x)$ from $\nabla^2F_{I(x)}$. This choice saves us from computing the Hessians of $f_i$, $i\in I(x)$ (notice that in case of $f_i(x)=a_i'x-b_i$ this makes no difference). However, this is a good choice only if we know that $C$ is nonempty. In this case for any $\bar{x}\in C$ we have $f_i(\bar{x})=0$ for $i\in I(\bar{x})$ and the term $\nabla^2 f_i(\bar{x}) f_i(\bar{x})$ vanishes. 

Algorithm~\ref{al:PSN_FEAS} is a regularized piecewise smooth Newton method with line search. Its convergence properties can be inferred as a special case of~\cite{PB13,PSB14}. Specifically, every accumulation point of the sequence generated is a stationary point of $F$, and if $\nabla^2F_{I(x^\star)}$ is nonsingular then the convergence rate is quadratic.
\begin{algorithm}
\label{al:PSN_FEAS}
\LinesNumbered
\DontPrintSemicolon
\caption{\mbox{[\texttt{empty}, $x$]=\texttt{PSN\_FEAS}$(C)$}} 
\KwIn{$\sigma\in \left(0,1/2\right)$, $\zeta\in (0,1)$, $x^0\in\rr^n$, $k=0$}
\uIf{$F(x^k)= 0$}{\texttt{empty}$\gets$\texttt{false}, $x\gets x^k$; \textbf{exit}}
\ElseIf{$\|\nabla F(x^k)\|=0$}{\texttt{empty}$=$\texttt{true}; \textbf{exit}}
Compute $d^ k\in\rr^n$ that solves
\begin{equation}\label{eq:RegNewtSys}
(\nabla^2F_{I^k}(x^k)+\delta^k I)d=-\nabla F(x^ k),
\end{equation}
where 
$I^k=\{i\in [m]\ |\ f_i(x^k)\geq 0\}$,
$\delta^k=\zeta\|\nabla F(x^ k)\|$.
\;
Compute 
$\tau^k=\max\{2^{-i}\ |\ i=0,1,2,\ldots\}$ such that
\begin{equation}\label{eq:Armijo}
F({x}^{ k}+\tau^k d^ k)\leq F({x}^{ k})+\sigma \tau_ k\nabla F({x}^{ k})'d^{ k}.
\end{equation}\; 
${x}^{ k+1}\gets {x}^{ k}+\tau^k d^ k$\;
$ k\gets k+1$ and go to Step 1.\;
\end{algorithm}

\subsection{Feasibility-based optimization}
Problems involving constrained minimization of a $\mathcal{C}^2$ convex function can be attacked by solving a sequence of feasibility problems. Specifically consider the problem
\begin{subequations}\label{eq:OptProb}
\begin{align}
\min&\ f_0(x)\\
\st&\ x\in C
\end{align}
\end{subequations}
where $C\subseteq\rr^m$ is a  closed convex set described as in~\eqref{eq:Feas}. Let $f^\star=\inf_{x\in C}f_0(x)$ and $X^\star=\arg\min_{x\in C}f_0(x)$.
We assume that
the set of optimal solutions of~\eqref{eq:OptProb} is nonempty. 
We have that
\begin{align*}
f^\star=\inf\{t:\ x\in S(t)\},\ X^\star=S(f^\star)
\end{align*}
where 
\begin{equation}\label{eq:levSet}
S(t)=\{x\in C:\ f_0(x)\leq t\}
\end{equation} 
is the lower level set of $f_0$ over $C$. Obviously we have that $t\geq f_\star$ if and only if $S(t)$ is nonempty. This suggests that we can test whether a given $t$ is smaller or larger than $f_\star$ by solving the feasibility problem of finding $x\in S(t)$
\begin{equation}\label{eq:phi}
\varphi(t)=\inf_{x\in\rr^n}\ \gamma(t,x)\eqdef\tfrac{1}{2}(f_0(x)-t)^2_++\tfrac{1}{2}\sum_{i=1}^m(f_i(x)_+)^2
\end{equation}
using Algorithm~\ref{al:PSN_FEAS}. The following proposition, whose proof is omitted here
for lack of space, proves some interesting properties enjoyed by function $\varphi$.
\begin{proposition}\label{prop:phiProps}
Assume problem~\eqref{eq:OptProb} admits an optimal solution and let function $\varphi$ be
defined as in~\eqref{eq:phi}.
\begin{enumerate}[(i)]
\item $\varphi$ is real-valued with $\varphi(t)>0$ for $t<f^\star$, whereas $\varphi(t)=0$ for $t\geq f^\star$,\label{it:welldef}
\item $\varphi$ is convex and continuously differentiable with 
\begin{equation}\label{eq:phiGrad}
\frac{d\varphi(t)}{dt}=\nabla_t\gamma(t,x_t)=-(f_0(x_t)-t)_+,
\end{equation}
where $x_t\in\arg\min_{x\in\rr^n}\gamma(x,t)$.
\item $\frac{d\varphi(t)}{dt}<0$ for $t<f^\star$, whereas $\frac{d\varphi(t)}{dt}=0$ for $t\geq f^\star$.\label{it:gradphi}
\end{enumerate}   
\end{proposition}
Proposition~\ref{prop:phiProps}\eqref{it:welldef} shows that $f^\star$ is the left endpoint of the halfline $\{t\in\rr:\ \varphi(t)=0\}$. Therefore problem~\eqref{eq:OptProb} has been reduced to finding the leftmost zero of the one-dimensional, monotone decreasing function $\varphi:\rr\to\rr$. 
One way to find $f^\star$ is to apply bisection to $\varphi$. Starting from an initial closed interval $[t_{-},t_{+}]$ with $t_{-}\leq f^\star\leq t_{+}$ we pick the midpoint $t=(t_{-}+t_{+})/2$ and try to determine if the level set $S(t)$ is nonempty, by solving~\eqref{eq:phi} using Algorithm~\ref{al:PSN_FEAS} (of course we could apply any other algorithm for unconstrained $\mathcal{C}^1$ optimization). If $\varphi(t)=0$ then $S(t)$ is nonempty and this means that the corresponding $x_t\in\arg\min_{x\in\rr^n}\gamma(t,x)$ is feasible for~\eqref{eq:OptProb}, therefore $t\geq f^\star$ and the new interval is reduced to $[t_{-},t]$. In the case where $\varphi(t)>0$, $S(t)$ is empty, meaning $t\leq f^\star$, so the new interval becomes $[t,t_{+}]$. 

\subsubsection{Strengthening the lower bound}
Overall, the bisection algorithm maintains a lower and upper bound for $f^\star$. Since the interval is halved at every bisection step, we obtain the standard linear convergence for bisection, that is, the algorithm stops after at most $\left\lceil\log_2\left(\frac{t_+-t_{-}}{\epsilon}\right)\right\rceil$ steps, where $\epsilon>0$ is the desired optimality threshold.
However, with almost no extra effort we can do much better in practice.

Suppose that $t<f^\star$.
The optimality condition for problem~\eqref{eq:phi} is $\nabla_x \gamma(x_t,t)=0$ or
\begin{equation}\label{eq:optphi}
(f_0(x_t)-t)_+\nabla f_0(x_t)+\sum_{i=1}^m(f_i(x_t))_+\nabla f_i(x_t)=0.
\end{equation}
From Proposition~\ref{prop:phiProps}\eqref{it:gradphi}, we have that $f(x_t)>t$.
Dividing by $f_0(x_t)-t$ in~\eqref{eq:optphi} 
and letting $\mu_i=\frac{(f_i(x_t))_+}{f_0(x_t)-t}$, $i=1,\ldots,m$,
we obtain
$$\nabla f_0(x_t)+\sum_{i=1}^m\mu_i\nabla f_i(x_t)=0.$$
Since $\mu\geq 0$, it follows that $\mu$ is a dual feasible vector, therefore
\begin{equation}\label{eq:LBdual}
f_0(x_t)+\sum_{i=1}^m\mu_i f_i(x_t)\leq f^\star.
\end{equation}
Hence $t_D=f_0(x_t)+\sum_{i=1}^m\mu_i f_i(x_t)$ provides a lower bound on $f_\star$.
Since $t_D-t= f_0(x_t)+\sum_{i=1}^m\mu_i f_i(x_t)>0$,  $t_D$ is indeed a tighter lower bound to $f^\star$ than $t$.
\subsubsection{Strengthening the upper bound}
When $t\geq f^\star$, due to~\ref{prop:phiProps}\eqref{it:gradphi}, we have that $f(x_t)\leq t$, $x_t\in\arg\min_{x\in\rr^n}\gamma(x,t)$. Therefore, if $\varphi(t)=0$, $x_t$ is a feasible
vector for problem~\eqref{eq:OptProb} and $f(x_t)$ is a tighter upper bound to $f^\star$ than $t$. 

In fact, we can do even better.
At every step of bisection we have at our disposal two vectors $x_F$ and $x_I$, corresponding to the upper and lower bounds on $f^\star$, respectively. Vector $x_F$ is feasible, i.e., $f_i(x_F)\leq 0$, $i=1,\ldots,m$, while $x_I$ is infeasible, i.e., $f_i(x_I)>0$ for at least one $i$ and $f_0(x_I)<f_0(x_F)$ (this follows directly from~\eqref{eq:LBdual}). Invoking  a result by Bertsekas~\cite[Proposition 2]{Ber99} we have that
$$f_\star\leq\frac{\Gamma}{\Gamma+1}f_0(x_F)+\frac{1}{\Gamma+1}f_0(x_I)\leq f_0(x_F),$$
where $\Gamma=\inf\{\gamma\geq 0\ |\ f_i(x_I)\leq-\gamma f_i(x_F),\ i=1,\ldots,m\}$~\cite[Proposition 2]{Ber99}. The bound is nontrivial, i.e., the rightmost inequality is strict, when $\Gamma<\infty$, which holds if and only if $f_i(x_I)\leq 0\quad\textrm{for all } i\textrm{ with }f_i(x_F)=0$.
In that case we have $\Gamma=\max_{\{i| f_i(x_F)<0\}}\frac{f_i(x_I)}{-f_i(x_F)}$.

The proposed improved bisection method is summarized in Algorithm~\ref{al:PSN_OPT}. 
\begin{algorithm}
\LinesNumbered
\DontPrintSemicolon
\caption{\texttt{PSN\_OPT}$(f_0,C)$}\label{al:PSN_OPT}
\KwIn{accuracy $\epsilon$, $x_F\in C$, $t_+=f_0(x_F)$, $x_I\notin C$, $t_{-}=f_0(x_I)$ with $t_{-}\leq f_\star$}
\KwOut{$x_F\in C$ with $f(x_F)-f^\star\leq\epsilon$}
\While{$t_+-t_{-}>\epsilon$}{
$t\gets(t_{-}+t_{+})/2$\;
Call [\texttt{empty},$x_t$]$=$\texttt{PSN\_FEAS}$(S(t))$ (Algorithm~\ref{al:PSN_FEAS})\;
\uIf{\texttt{empty}$=$\texttt{true}}{
$x_I\gets x_t$\;
$t_{-}\gets f_0(x_I)+\sum_{i=1}^m\mu_i f_i(x_I)$, 
where $\mu_i=\frac{(f_i(x_I))_+}{f_0(x_I)-t}$\;
\If{$f_i(x_I)\leq 0\quad\textrm{for all } i\textrm{ with }f_i(x_F)=0$}{
$t_{+}\gets \frac{\Gamma}{\Gamma+1}f_0(x_F)+\frac{1}{\Gamma+1}f_0(x_I)$, where $\Gamma=\max_{\{i| f_i(x_F)<0\}}\frac{f_i(x_I)}{-f_i(x_F)}$}
}
\Else{
$x_F\gets x$\;
\uIf{$f_i(x_I)\leq 0\quad\textrm{for all } i\textrm{ with }f_i(x_F)=0$}{
$t_{+}\gets \frac{\Gamma}{\Gamma+1}f_0(x_F)+\frac{1}{\Gamma+1}f_0(x_I)$, where $\Gamma=\max_{\{i| f_i(x_F)<0\}}\frac{f_i(x_I)}{-f_i(x_F)}$}
\Else{$t_+\gets f_0(x_F)$}
}
}
\end{algorithm}

\subsubsection{Equality constraints}
The approach can be immediately extended to handle linear equality constraints 
\[
	c_j'x=d_j,\ j=1,\ldots,m_e
\]
in~\eqref{eq:Feas}, by simply adding the term $\frac{1}{2}\sum_{j=1}^{n_e}(c_j'x-d_j)^2$ in~\eqref{eq:FeasMin} and to $\varphi(t)$ in~\eqref{eq:phi}, respectively.

\subsubsection{Determining initial upper and lower bounds}
To determine an initial upper bound $t_{+}$ to $f^\star$ we can solve the feasibility problem~\eqref{eq:Feas} using Algorithm~\ref{al:PSN_FEAS}. Determining a lower bound is a more delicate issue. If $f_0(x)=(1/2)x'Qx+q'x$ where $Q$ is symmetric positive definite, we can simply determine a lower bound on $f^\star$ by computing the unconstained minimum $x=-Q^{-1}q$. In general, if $f_0$ is convex and coercive we can find a $x\in\rr^n$ such that $\nabla f(x)=0$. The nonlinear system can be solved by Algorithm~\ref{al:PSN_FEAS}. In the case of a quadratic program with the cost having a positive semidefinite Hessian a lower bound to $f^\star$ can be determined by solving the following convex feasibility problem
$$\textrm{find } x\in\rr^n, \mu\in\rr^m,\st \nabla f_0(x)+\sum_{i=1}^m\mu_i\nabla f_i(x)=0, \mu\geq 0$$  
Then $\mu$ is a dual feasible solution and $q=f_0(x)+\sum_{i=1}^m\mu_if_i(x)\leq f^\star$ (even if strong duality does not hold).
Another way to determine a lower bound for general convex problems is to find a dual feasible vector, corresponding to the primal feasible vector $x_F$ corresponding to $t_{+}$:
$$\textrm{find } \mu\in\rr^m,\st \nabla f_0(x_F)+\sum_{i=1}^m\mu_i\nabla f_i(x_F)=0,\ \mu\geq 0.$$
Notice that the set of $\mu$ satisfying the conditions above is polyhedral.

\subsubsection{Special cases}
Algorithm~\ref{al:PSN_FEAS} and~\ref{al:PSN_OPT} can be used to solve linear programs (LPs), quadratic programs (QPs), and quadratically-constrained quadratic programs (QCQPs). In this case, the computations in~\eqref{eq:Nabla_F} and~\eqref{eq:Nabla2_F} require only matrix-vector products.

\subsection{Applicability to feasibility-based MPC} 
Algorithm~\ref{al:PSN_OPT} requires all the constraints in the inner feasibility problem to 
be twice differentiable functions. In particular, constraint~\eqref{eq:stab-constraint} is not 
continuously differentiable because of the $\max$ operator, so the above algorithm cannot be directly applied
to solve~\eqref{eq:MPC-problem}. However, we can simply recast the problem by introducing $N-1$ additional variables~$\epsilon_k$, $k=0,\ldots,N-1$ and replace~\eqref{eq:stab-constraint}
with
\begin{subequations}
\begin{IEEEeqnarray}{lcl}
     &\ & \epsilon_k\geq f(x_k),\ k=1,2,\ldots,N-1\label{eq:epsilon2}\\
     &\ & \epsilon_k\geq 0\\
     &\ & \sum_{k=1}^{N-1} \epsilon_k\leq \phi(t)-f(x(t))_+\label{eq:stab-constraint2}
\end{IEEEeqnarray}%
\end{subequations}
\label{eq:MPC-problem2}%
without altering feasibility and optimality of the solutions.

Moreover, in case $f$ is given as the max of convex functions $f_i\fundef{n}{}$, $i=1,\ldots,n_f$, constraints~\eqref{eq:term-constraint} and~\eqref{eq:epsilon2} can be replaced by
\begin{subequations}
\begin{IEEEeqnarray}{lcl}
     &\ & f_i(x_N)\leq 0,\ i=1,\ldots,n_f\label{eq:term-constraint3}\\
     &\ & \epsilon_k\geq f_i(x_k),\ i=1,\ldots,n_f,\ k=1,\ldots,N-1\label{eq:epsilon3}
\end{IEEEeqnarray}%
\end{subequations}
Similarly, if $g$ is given as the max of convex functions $g_i\fundef{n}{}$, $i=1,\ldots,n_g$,~\eqref{eq:stage-constraint} can be replaced by
\begin{equation}
g_i(x_k,u_k)\leq 0,\ i=1,\ldots,n_g,\ k=0,\ldots,N-1
\label{eq:stage-constraint3}
\end{equation}
The case $\ell_k=\max_{i=1,\ldots,n_{\ell k}}\ell_{ik}$ can also be dealt with by introducing
$\sum_{k=0}^Nn_{\ell k}$ additional variables $\sigma_{ik}$, $k=0,\ldots,N$, and replacing~\eqref{eq:cost} with
\begin{subequations}
\begin{align}
\min \ & \sum_{k=0}^N\sum_{i=1}^{n_{\ell k}}\sigma_{ik} \label{eq:cost3}\\
     \st\ & \sigma_{ik}\geq \ell_{ik}(x_k,u_k),\, i=1,\ldots,n_{\ell k},\,k=0,\ldots,N-1\label{eq:epsilon3}\\
     & \sigma_{iN}\geq \ell_{iN}(x_k),\ i=1,\ldots,n_{\ell N}\label{eq:epsilon31}
\end{align}
\end{subequations}
In conclusion, Algorithm~\ref{al:PSN_OPT} can be applied to solve~\eqref{eq:MPC-problem}
for any twice differentiable convex function $f_i$, $g_i$, $\ell_{ik}$.

\section{Constrained tracking to a set}
\label{sec:terminal-set}
Consider an output vector
\begin{equation}
	y(t)=Cx(t)
\label{eq:process-output}
\end{equation}
associated with process~\eqref{eq:process}, with $y\in\rr^p$, and
a corresponding output reference $r\in\RR\subset\rr^p$, 
where $\RR$ is a polytope. Assume that the following linear system
\[
	\ba{rcl}
	0&=&Ax_r+Bu_r-x_r\\
	r&=&Cx_r
	\ea
\]
admits a unique solution $(x_r,u_r)$ of steady-state state and input
vectors for all $r\in\RR$, 
and assume that $\RR$ is such that $g(x_r,u_r)<0$.

We consider the problem of controlling~\eqref{eq:process} to a desired set
$\XX_T\eqdef\{x\in\rr^n:\ S(x-x_r)\leq s\}$,
around the equilibrium state $x_r$, where $s>0$, $s\in\rr^{n_T}$, 
while satisfying the input constraints
\begin{equation}
\umin\leq u(t)\leq \umax
\label{eq:input-constraints}
\end{equation}
and the output constraints
\begin{equation}
\ymin \leq y(t) \leq \ymax
\label{eq:output-constraints}
\end{equation}
with $\umin<0<\umax$, $\ymin<0<\ymax$.
In this case, we define the function $g$ in~\eqref{eq:stage-constraint} as the convex and piecewise affine function
\begin{equation}
g(x,u)=\max_{i=1,\ldots,q}\{G^x_i(Ax+Bu)+G^u_iu-g^0_i\}
\label{eq:linear-constraints}
\end{equation}
where $G^x=\smallmat{0\\0\\ C\\-C}$, $G^u=\smallmat{I\\-I\\0\\0}$,
$g^0=\smallmat{\umax\\-\umin\\\ymax\\-\ymin}$, and $q\eqdef 2m+2p$.
An example of desired set $\XX_T$ is given by $S=\smallmat{C\\-C}$, $s=\smallmat{\emax\\-\emin}$, so that convergence to $\XX_T$ implies satisfying
the constraint on the tracking error $\emin \leq y-r \leq \emax$ asymptotically.

\subsection{Quadratic functions}
We consider the ellipsoidal terminal set $\Ss$ defined by
\begin{equation}
	f(x)=(x-x_r)'P(x-x_r)-\rho_r
\label{eq:quadratic_f(x)}
\end{equation}
where $P=P'\geq 0$ and $\rho_r$ are determined in accordance with the following theorem.

\begin{theorem}
Let $Q=Q'\geq 0$, $Q\in\rr^{n\times n}$, $Y\in\rr^{m\times n}$,
$X=X'\geq 0$, $X\in\rr^{m\times m}$, be the solution of the semidefinite program
\begin{subequations}
\begin{IEEEeqnarray}{lcl}
\max &\ & (\det Q)^\frac{1}{n} \label{eq:LMI-cost}\\
\st  &\ & \matrice{cc}{Q & (AQ+BY)'\\ AQ+BY & \lambda Q}\geq 0\label{eq:LMI-contractive}\\
     &\ & \matrice{cc}{X & Y\\ Y' & Q}\geq 0 \label{eq:LMI-input-constraint-1} \\
     &\ & X_{ii}\leq \bar u_i^2,\ i=1,\ldots,m \label{eq:LMI-input-constraint-2}\\
     &\ & \matrice{cc}{Q & (AQ+BY)'C_i'\\ C_i(AQ+BY) & \bar {y}_i^2}
     \label{eq:LMI-output-constraint}\\&\ & ~~~i=1,\ldots,p\nonumber\\
     &\ & \matrice{cc}{Q & Q'S_i'\\
         S_iQ & s_i^2},\ i=1,\ldots,n_T\label{eq:LMI-terminal-constraint}
\end{IEEEeqnarray}%
\label{eq:LMI}%
\end{subequations}
where $0\leq \lambda\leq 1$ is a given contractive factor, and 
\begin{subequations}
\beqar
\bar u_i\eqdef\min_{j=1,\ldots,n_R}\{u_{\rm max,i}-(u_{r,j})_i,-u_{\rm min,i}+(u_{r,i})_i\}
\label{eq:umax} \\
\bar y_i\eqdef\min_{j=1,\ldots,n_R}\{y_{\rm max,i}-(r_{j})_i,-y_{\rm min,i}+(r_{j})_i\}
\label{eq:ymax}
\eeqar%
\label{eq:umax-ymax}%
\end{subequations}
for all $i=1,\ldots,q$, where $\{r_j\}_{j=1}^{n_R}$ are the vertices of $\RR$. 
Then, by setting $K=YQ^{-1}$, $P\eqdef Q^{-1}$, $f(x)=(x-x_r)'P(x-x_r)-\rho_r$, 
the set $\Ss=\{x: (x-x_r)'P(x-x_r)\leq \rho_r\}$ is constrained controlled invariant
under the control law $u=K(x-x_r)+u_r$ in that $x\in\Ss$ implies
\begin{subequations}
\begin{IEEEeqnarray}{c}
Ax+Bu\in\Ss\\
\umin\leq K(x-x_r)+u_r\leq \umax\\
\ymin\leq C(Ax+Bu)\leq \ymax \\
S x\leq s
\end{IEEEeqnarray}%
\label{eq:LMI-properties}%
\end{subequations}
for all $r\in\RR$, where
\begin{subequations}
\begin{equation}
\rho_r=\min_{i=1,\ldots,q+n_T}\left\{
\frac{\bar b_i^2}{\bar A_iQ\bar A_i'}
\right\}\geq 1
\label{eq:rho_r}
\end{equation}
\begin{equation}
	\bar A=\matrice{c}{K\\-K\\C(A+BK)\\-C(A+BK)\\S},\ 
    \bar b=\matrice{c}{\umax-u_r\\-\umin+u_r\\\ymax-r\\-\ymin+r\\s}
\label{eq:bar_A_bar_b}
\end{equation}%
\label{eq:rho_inflated}%
\end{subequations}
\end{theorem}

\proof Let $\Delta x\eqdef x-x_r$, $\Delta u\eqdef u-u_r$, 
$\Delta y\eqdef y-r$. Clearly $\Delta x(t+1)=Ax+Bu-x_r=A\Delta x(t)+B\Delta u(t)$ and $\Delta y(t)=C\Delta x(t)$, along with the constraints 
$\umin-u_r\leq \Delta u(t)\leq \umax-u_r$,
$\ymin-r\leq C\Delta x(t)\leq \ymax-r$. 
The robust satisfaction of properties~\eqref{eq:LMI-properties} 
with respect to $r\in\RR$ for all $x$ such
that $\Delta x'P\Delta x\leq 1$, under the control law $\Delta u(t)=K\Delta x(t)$, follows by
standard arguments from the inequality constraints in~\eqref{eq:LMI} (see, e.g.,~\cite{KBM96}).
For a given $r\in\RR$, the scalar $\rho_r$ defined by~\eqref{eq:rho_inflated} provides the largest ellipsoid centered in $x_r$ and defined by $P$ 
such that~\eqref{eq:LMI-properties} are satisfied, where $\rho_r\geq 1$.
\endproof

\subsection{Polyhedral terminal set $\Ss$}
For a given under asymptotically stabilizing feedback control law $u(t)=K(x(t)-x_r)+u_r$,
consider the polyhedral terminal set $\Ss$ defined by
\begin{equation}
	f(x)=\max\{H_i'(x-x_r)-K_i\}
\label{eq:PWA_f(x)}
\end{equation}
where $\Ss=\{x:\ H(x-x_r)\leq K\}=\{\Delta x:\ H\Delta\leq K\}$ is a maximum admissible polyhedral invariant set~\cite{GT91} for the closed-loop system $\Delta x(t+1)=(A+BK)\Delta x(t)$
and with respect to the constraints $\bar A\Delta x\leq \bar b_{\rm min}$, where $\bar A$ is
defined in~\eqref{eq:bar_A_bar_b}, and $\bar b_{\rm min}$ is defined as 
in~\eqref{eq:bar_A_bar_b} by replacing $(\umax-u_r)_i$
with $\min_{j=1,\ldots,n_R}\{(\umax-u_{r_j})_i\}$,
$(-\umin+u_r)_i$
with $\min_{j=1,\ldots,n_R}\{(-\umin+u_{r_j})_i\}$, $i=1,\ldots,m$, 
$(\ymax-r)_i$
with $\min_{j=1,\ldots,n_R}\{(\ymax-r_j)_i\}$ and
$(-\ymin+r)_i$
with $\min_{j=1,\ldots,n_R}\{(-\ymin+r_j)_i\}$, $i=1,\ldots,p$. Clearly, 
the size of $\Ss$ depends on the size of $\XX_T$ and on how large are
the boxes $\{y\in\rr^p:\ \ymin\leq y\leq \ymax\}$ with respect to $\RR$ and 
$\{u\in\rr^m:\ \umin\leq u\leq \umax\}$ with respect to the set $\{u\in\rr^m:\ u=u_r,\
r\in\RR\}$.

\section{Simulation Results}
\label{sec:simulations}

Consider the linear system described by the transfer function
\begin{equation}
	G(s)=\frac{26(s+1)}{s^2+2s+26}
\label{eq:example-model}
\end{equation}
Model~\eqref{eq:example-model} is converted to discrete-time by
exact sampling plus zero-order hold with sampling time $T_s=0.2$~s, 
resulting in the state-space model with matrices $A=\smallmat{ 0.4424  &1\\
   -0.4746  &  0.4424}$, $B=\smallmat{0\\2.0623}$, $C=\smallmat{
   -0.7013 &  1.9407}$, $D=0$. The system is 
subject to the constraints
\begin{equation}
	-1\leq u\leq 1,\ -1\leq y\leq 1
\label{eq:example-IO-constraints}
\end{equation}
leading to defining $g(x,u)$ as in~\eqref{eq:linear-constraints}.
We setup the MPC problem~\eqref{eq:MPC-problem} with $N=6$,
\begin{equation}
	\ell_k(x,u)=(x-x_r)'Q(x-x_r)+(u-u_r)'R(x-x_r)
\label{eq:example-stage-cost}
\end{equation}
$Q=10C'C$, $R=1$, for all $k=0,\ldots,N-1$, and 
\begin{equation}
	\ell_N(x)=(x-x_r)'P(x-x_r)
\label{eq:example-stage-cost}
\end{equation}
where $P$ is the solution of the discrete algebraic
Riccati equation associated with $A$, $B$, $Q$, and $R$,
and $x_r=\smallmat{2.6252\\1.4639}r$, $u_r=r$. The possible reference 
signals are restricted in the interval $\RR=[-0.9,0.9]$,
while the desired target set $\XX_T=\{x\in\rr^2:\ \|x-x_r\|_\infty\leq 0.1\}$.
We start from the initial condition $x(0)=\smallmat{0\\0}$ and command the set-point $r=0.5$.
We consider the cumulated cost 
\[
	J\eqdef\sum_{t=0}^{10}\big(\ell_N(x_N^t)+\sum_{k=0}^{N-1}\ell_k(x_k^t,u_k^t)\big)
\]
as a measure of closed-loop performance, where $\{u_k^t\}_{k=0}^{N-1}$ is the solution of problem~\eqref{eq:MPC-problem} chosen at time $t$, and $\{x_k^t\}_{k=0}^{N}$ the corresponding state trajectory.
We consider two settings for function $f$: 

Case ($i$): the convex quadratic function as in~\eqref{eq:quadratic_f(x)}, where $P=\smallmat{ 105.4493 & 23.9713\\  23.9713 & 105.4493}$ is obtained by~\eqref{eq:LMI}
with $\lambda=1$, along with the terminal gain $K=\smallmat{  0.1968 &  -0.2898}$.
Table~\ref{tab:ellipse} 
shows the cumulated cost $J$ for different maximum values of permitted CPU time to solve problem~\eqref{eq:MPC-problem}. The corresponding trajectories are depicted in Figures~\ref{fig:plots-ellipse},~\ref{fig-state-space-ellipse}.

\begin{table}[h]
\begin{center}
\caption{Performance vs. allocated CPU time (ellipsoidal constraints)}
\label{tab:ellipse}
\begin{tabular}{|c|c|}
\hline
max CPU time (ms) & Cumulated cost $J$\\\hline\hline
 unbounded &   8.7865\\\hline
 60 &  22.6572\\\hline
 40 &  26.4663\\\hline
 20 &  31.4528\\\hline
\end{tabular}
\end{center}
\end{table}

Case ($ii$): the convex piecewise affine function 
as in~\eqref{eq:PWA_f(x)}, where $H=\smallmat{ 1 & 0 & -1 & 0 & 0.4424 &  -0.4424\\
0 & 1 & 0 & 01 & 1 & -1}'$,
$K=0.1\smallmat{1&1&1&1&1&1}'$ is the maximum $\lambda$-contractive invariant set
for the closed-loop system $\Delta x(t+1)=(A+BK_{\rm LQR})\Delta x(t)$,
$K_{\rm LQR}$ is the LQR gain associated with $A$, $B$, $Q$, and $R$, and $\lambda=1$,
computed as described in~\cite{BOPS11}. Table~\ref{tab:polyhedral} 
shows the cumulated cost $J$ for different maximum values of permitted CPU time to solve problem~\eqref{eq:MPC-problem}. The corresponding trajectories are depicted in Figures~\ref{fig:plots-polyhedral},~\ref{fig-state-space-polyhedral}.

\begin{table}[h]
\begin{center}
\caption{Performance vs. allocated CPU time (polyhedral constraints)}
\label{tab:polyhedral}
\begin{tabular}{|c|c|}
\hline
max CPU time (ms) & Cumulated cost $J$\\\hline\hline
     unbounded &   9.0075\\\hline
      10 &  23.4491\\\hline
       5 &  26.4646\\\hline
       1 &  26.7551\\\hline
\end{tabular}
\end{center}
\end{table}

Finally, we compare the performance of the new solver described in Section~\ref{sec:algo}
(implemented in interpreted MATLAB code) against the commercial solver Gurobi 5.6.2~\cite{gurobi14} in solving QCQP and QP problems
to optimality, and also to qpOASES~\cite{QPOASES} for QP's. We consider problems deriving from the MPC setup with ellipsoidal (QCQP) 
and polyhedral (QP) constraints described above, for an increasing prediction horizon $N$.
The results are depicted in Figure~\ref{fig:CPUtime-ellipse} (QCQP case) and Figure~\ref{fig:CPUtime-polyhedron} (QP case), respectively.

\section{Conclusions}
The contribution of this paper is twofold. From an optimization viewpoint, we have introduced
a very efficient numerical solver that can solve convex feasibility and optimization problems, and that is at least an order of magnitude faster than commercial state-of-the-art (interior-point) solvers as the dimension of the problem increases. By taking advantage of the way the solver computes an optimal solution via a sequence of convex feasibility problems, from a control viewpoint we proposed an MPC strategy for convergence to a terminal set that allows an \emph{anytime optimization} philosophy, that is of improving the optimality of the control
move with respect to a given performance specification only if CPU resources are available during the sampling interval. We believe that the approach has potential applications in embedded MPC systems where a large-enough time-slot for computations cannot be guaranteed a priori, a rather typical situation in multitask real-time systems. 

\begin{figure}[t!]
\begin{center}
\includegraphics[width=\hsize]{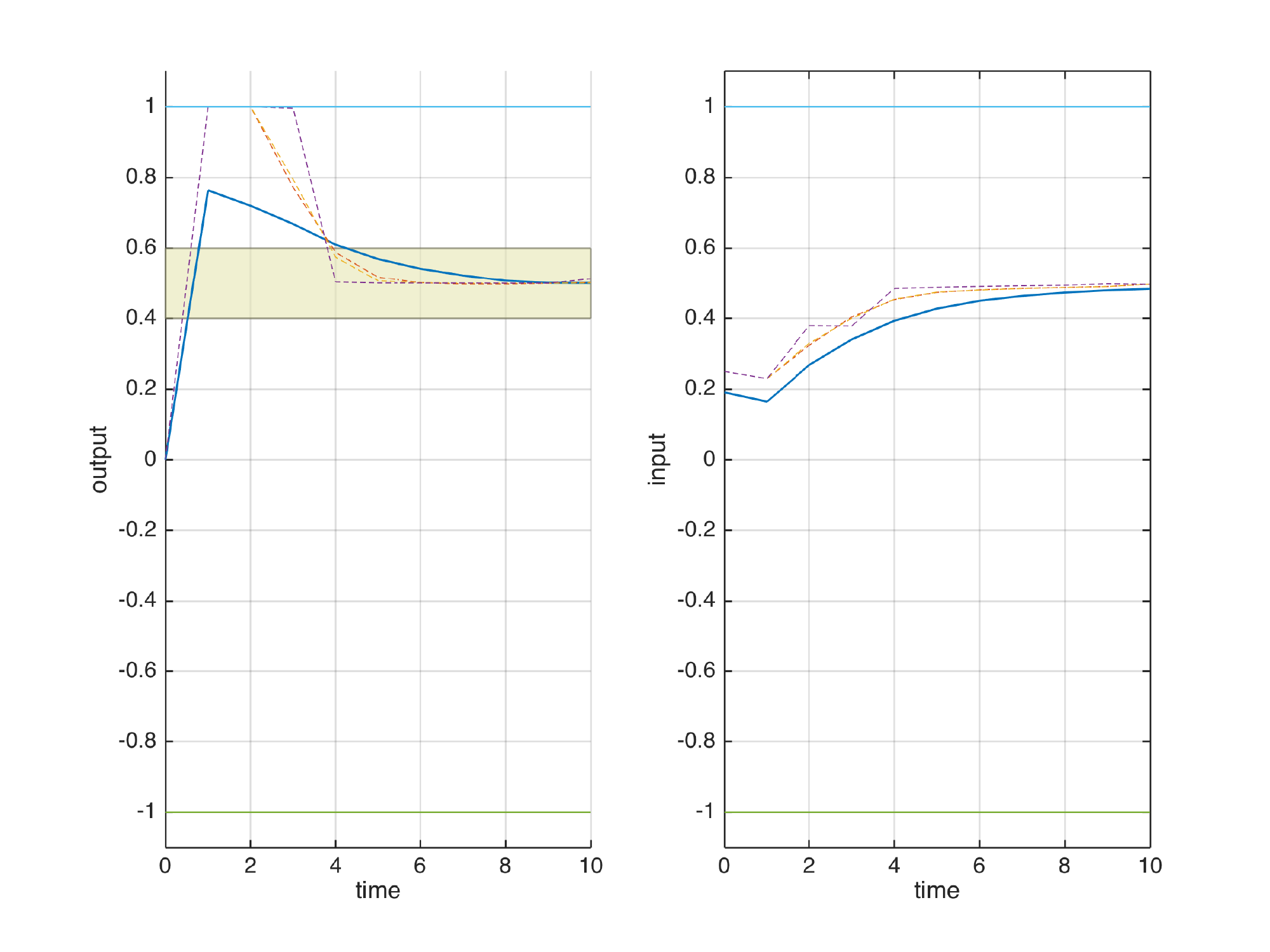}
\end{center}
\caption{System trajectories for varying values of available CPU time (solid=optimal),
ellipsoidal constraints}
\label{fig:plots-ellipse}
\end{figure}

\begin{figure}[t!]
\begin{center}
\includegraphics[width=\hsize]{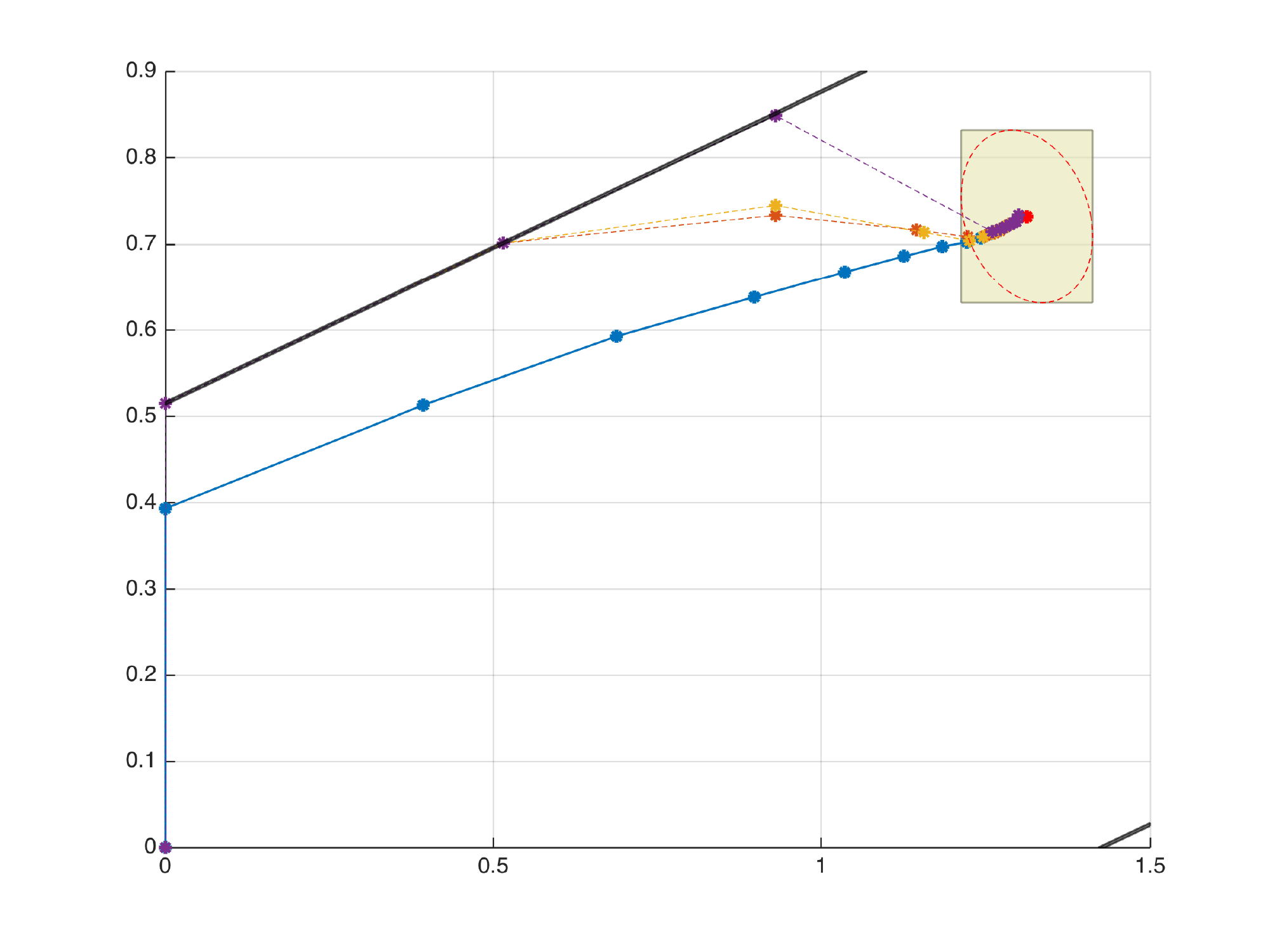}
\end{center}
\caption{State-space trajectories for varying values of available CPU time (solid=optimal),
ellipsoidal constraints}
\label{fig-state-space-ellipse}
\end{figure}


\begin{figure}[t!]
\begin{center}
\includegraphics[width=\hsize]{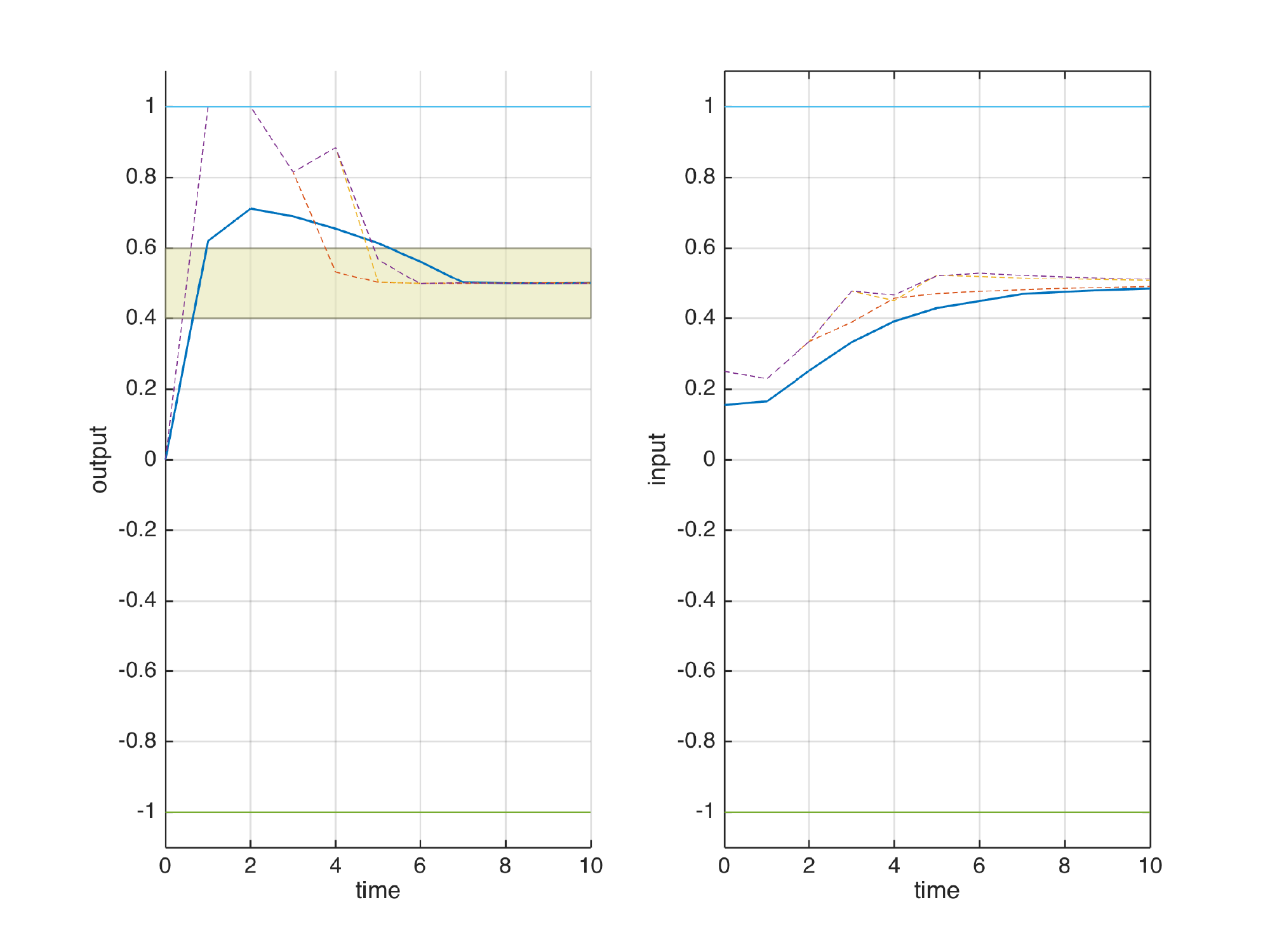}
\end{center}
\caption{System trajectories for varying values of available CPU time (solid=optimal),
polyhedral constraints}
\label{fig:plots-polyhedral}
\end{figure}

\begin{figure}[t!]
\begin{center}
\includegraphics[width=\hsize]{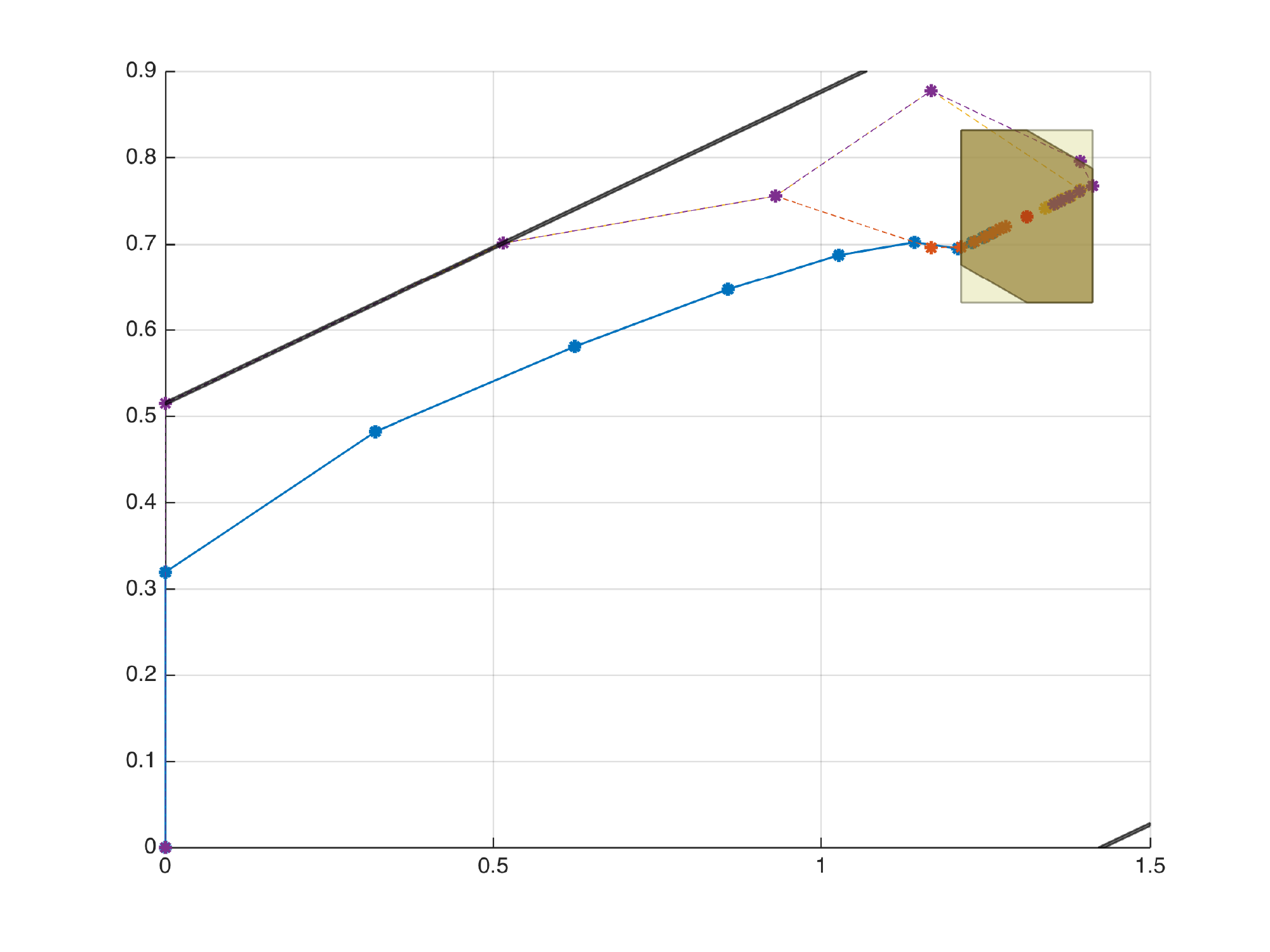}
\end{center}
\caption{State-space trajectories for varying values of available CPU time (boldface=optimal),
polyhedral constraints}
\label{fig-state-space-polyhedral}
\end{figure}

\begin{figure}[t!]
\begin{center}
\includegraphics[width=\hsize]{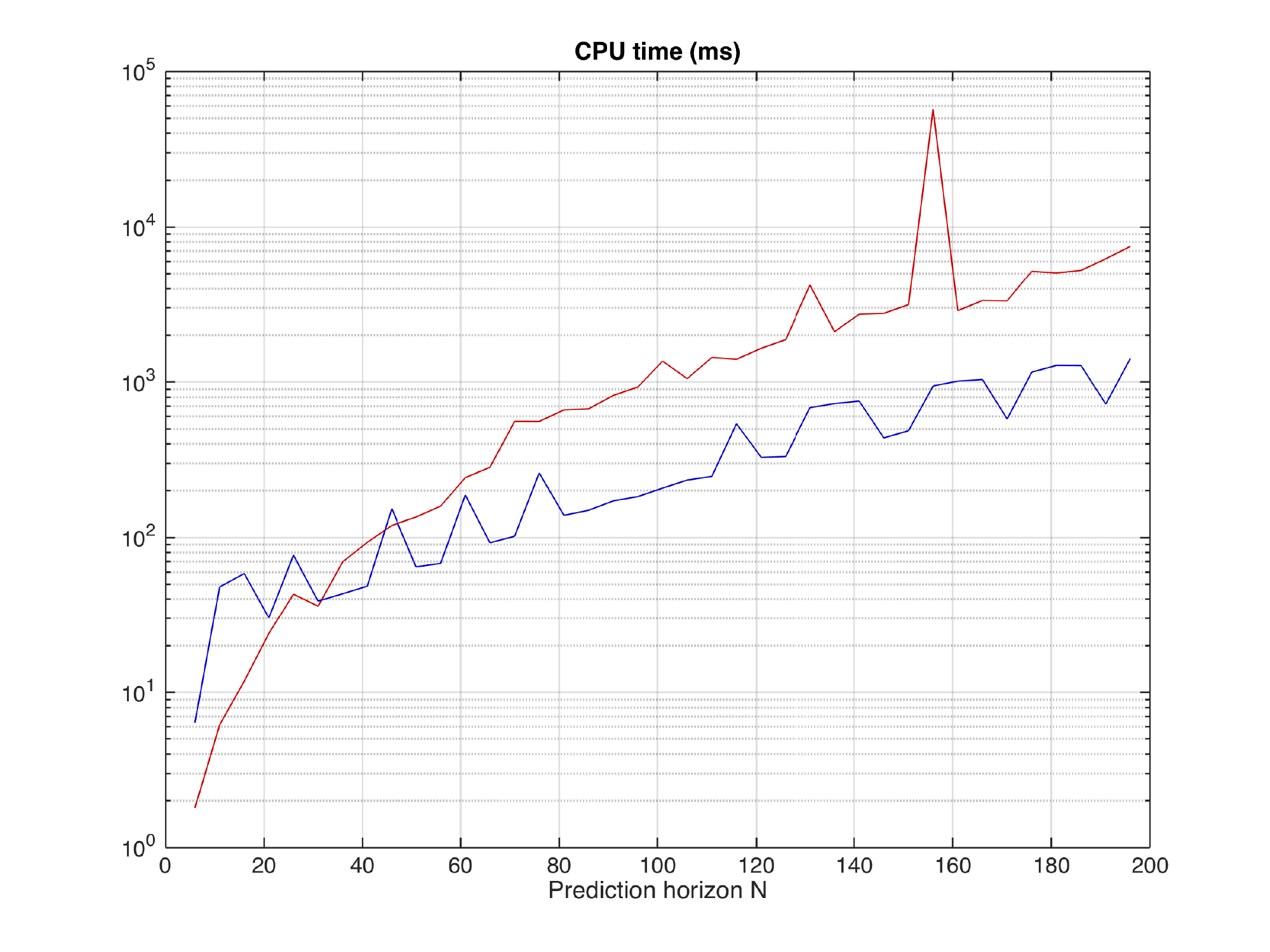}
\end{center}
\caption{QCQP: Algorithm~\ref{al:PSN_OPT} (blue) vs Gurobi (red)}
\label{fig:CPUtime-ellipse}
\end{figure}

\begin{figure}[t!]
\begin{center}
\includegraphics[width=\hsize]{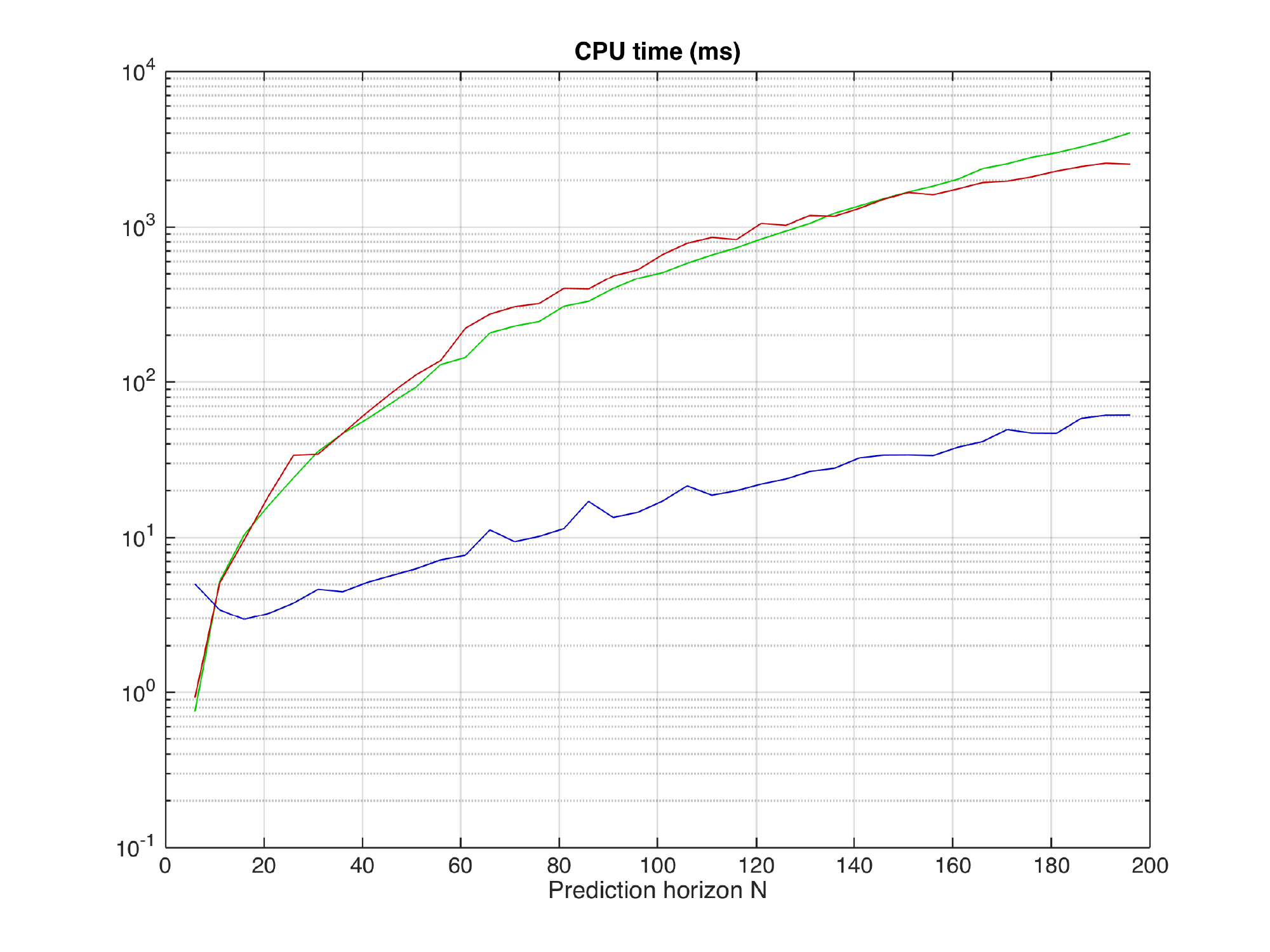}
\end{center}
\caption{QP: Algorithm~\ref{al:PSN_OPT} (blue), Gurobi (red), qpOASES (green)}
\label{fig:CPUtime-polyhedron}
\end{figure}

\end{document}